\documentclass[12pt]{iopart}
\usepackage{iopams} 
\usepackage{amstext}
\usepackage{amsthm}
\usepackage{dsfont}
\usepackage{hyperref}
\usepackage{enumitem}
\expandafter\let\csname equation*\endcsname=\relax
\expandafter\let\csname endequation*\endcsname=\relax
\usepackage{amsmath}
\usepackage{caption}
\usepackage{subcaption}

\usepackage{graphicx}

\newcommand{\diff}{{\rm d}}
\newcommand{\EE}{\mathbb{E}}
\newcommand{\torus}{\mathbb{T}_L^d}

\newcommand{\var}{\mathrm{var}}
\newcommand{\PP}{\mathbb{P}}
\newcommand{\NN}{\mathbb{N}}
\newcommand{\ZZ}{\mathbb{Z}}
\newcommand{\RR}{\mathbb{R}}
\newcommand{\posint}{\ZZ_{+}}
\newcommand{\naturals}{\NN}
\newcommand{\sA}{\mathcal{A}}
\newcommand{\sC}{\mathcal{C}}
\newcommand{\sK}{\mathcal{K}}
\newcommand{\sL}{\mathcal{L}}
\newcommand{\sN}{\mathcal{N}}
\newcommand{\sR}{\mathcal{R}}
\newcommand{\sS}{\mathcal{S}}
\newcommand{\sT}{\mathcal{T}}
\newcommand{\sW}{\mathcal{W}}
\newcommand{\sX}{\mathcal{X}}
\newcommand{\sZ}{\mathcal{Z}}
\newcommand{\dc}{d_{\rm c}}
\newcommand{\ind}{{\mathds{1}}}

\newcommand{\Phibar}{\bar{\Phi}}

\newtheorem{theorem}{Theorem}[section]
\newtheorem{lemma}[theorem]{Lemma}

\newtheorem{proposition}[theorem]{Proposition}

\newcommand{\dt}{\diff t}

\begin{document} 

\title{Unwrapped two-point functions on high-dimensional tori}

\bigskip \begin{center}In memory of Norman E. Frankel\end{center}

\author{Youjin Deng$^{1,2}$, Timothy M. Garoni$^3$, Jens Grimm$^3$, Zongzheng Zhou$^3$}
\address{$^1$Department of Modern Physics, University of Science and
  Technology of China, Hefei 230027, China}
\address{$^2$MinJiang Collaborative Center for Theoretical Physics,
College of Physics and Electronic Information Engineering, Minjiang University, Fuzhou 350108, China}
\address{$^3$ARC Centre of Excellence for Mathematical and Statistical Frontiers (ACEMS), School of Mathematics, Monash University, Clayton,
  Victoria 3800, Australia}
\ead{\mailto{tim.garoni@monash.edu}, \mailto{eric.zhou@monash.edu}, \mailto{yjdeng@ustc.edu.cn}}
  
\begin{abstract}
  We study \textit{unwrapped} two-point functions for the Ising model, the self-avoiding walk and a random-length loop-erased random walk on
  high-dimensional lattices with periodic boundary conditions. While the standard two-point functions of these models have been observed to
  display an anomalous plateau behaviour, the unwrapped two-point functions are shown to display standard mean-field behaviour. 
  Moreover, we argue that the asymptotic behaviour of these unwrapped two-point functions on the torus can be
  understood in terms of the standard two-point function of a random-length random walk model on $\ZZ^d$. A precise description is derived
  for the asymptotic behaviour of the latter. Finally, we consider a natural notion of the Ising \emph{walk length}, and show numerically that
  the Ising and SAW walk lengths on high-dimensional tori show the same universal behaviour known for the SAW walk length on the complete graph.
\end{abstract}

\noindent{\it Keywords}: Upper critical dimension, Finite-size scaling, Ising model, Self-avoiding walk, two-point function

\section{Introduction}
It is well known~\cite{FernandezFrohlichSokal1992} that models of critical phenomena typically possess an upper critical dimension $\dc$,
such that in dimensions $d>\dc$, their thermodynamic behaviour is governed by critical exponents taking simple mean-field values. In contrast
to the simplicity of the thermodynamic behaviour, however, the theory of finite-size scaling in dimensions above $\dc$ is surprisingly
subtle, and has been the subject of considerable debate; see
e.g.~\cite{LundowMarkstrom2014,WittmannYoung2014,LundowMarkstrom2016,FloresSola2016,GrimmElciZhouGaroniDeng2017,ZhouGrimmFangDengGaroni2018}. 

In particular, it has been observed that when $d>\dc$ the finite-size scaling of a number of fundamental quantities depends strongly on the boundary
conditions imposed. For example, for the Ising model and the self-avoiding walk at their infinite-volume critical points, it has been
numerically observed that on a box of linear size $L$ with free boundary conditions, the two-point function and susceptibility display the
expected mean-field behaviour, $g(x) \approx \| x\|^{2-d}$~\cite{ZhouGrimmFangDengGaroni2018} and $\chi \approx
L^2$~\cite{LundowMarkstrom2014,LundowMarkstrom2016,ZhouGrimmFangDengGaroni2018}, respectively. These observations have recently been verified
rigorously in the Ising case~\cite{CamiaJiangNewman2021}. By contrast, if periodic boundary conditions are imposed, i.e. the model is
defined on a discrete torus, then 
simulations~\cite{Binder1985,GrimmElciZhouGaroniDeng2017,ZhouGrimmFangDengGaroni2018} suggest the anomalous behaviour $g(x) \approx c_1 \| x
\|^{2-d} + c_2 L^{-d/2}$ and $\chi \approx L^{d/2}$ holds, as predicted for the Ising case in~\cite{Papathanakos2006}.
This so-called \emph{plateau} behaviour of the two-point function has recently been established rigorously~\cite{SladeWSAW2020} for the Domb-Joyce model with
$d>4$, for sufficiently weak interaction strength, and also for bond percolation~\cite{HutchcroftMichtaSlade2021}
when $d\ge11$ for the nearest-neighbour model, and $d>6$ for spread-out models.

In this article, we will focus solely on the case of periodic boundary conditions.
It was argued heuristically and observed numerically in~\cite{GrimmElciZhouGaroniDeng2017} that the expected number of windings of a SAW on
a torus of dimension $d>\dc$ should scale like $L^{d/\dc-1}$. This implies that there is a proliferation of windings when $d>\dc$.
Analogous behaviour has recently been established rigorously for bond percolation; indeed, it was proved in~\cite{HeydenreichHofstad2017} that, with
high probability, large clusters contain long cycles which wind the torus at least $L^{d/\dc-1}$ times.

In an effort to understand the plateau behaviour of the SAW/Ising torus two-point function, 
it was argued in~\cite{GrimmElciZhouGaroniDeng2017} that if one considers an alternative \emph{unwrapped} two-point function, 
which correctly accounts for the proliferation of windings, then the
standard mean-field behaviour is recovered in the bulk. Strong numerical evidence in support of this claim was presented for the case of
SAW. The unwrapping procedure described in~\cite{GrimmElciZhouGaroniDeng2017} was formulated in the language of walk models however, and no
analogous construction was provided for the Ising model. One contribution of the current article is to consider a natural
walk model associated with the Ising model~\cite{Aizenman1986,FernandezFrohlichSokal1992}, and use it to define an unwrapped analogue
of the Ising two-point function. As described below, this unwrapped two-point function displays the same asymptotic behaviour as in the SAW case.

In fact, by studying the random-length random walk (RLRW) introduced in~\cite{ZhouGrimmFangDengGaroni2018}, we make a rather more detailed
prediction for the behaviour of the Ising/SAW unwrapped two-point function than discussed in~\cite{GrimmElciZhouGaroniDeng2017}. Specifically, we
provide a concrete conjecture for its universal behaviour on the scale of the \emph{unwrapped length}, $L^{d/\dc}$. Strong numerical evidence,
provided by Monte Carlo simulation, is then provided in support of this conjecture. In addition to the Ising and SAW cases, we also present
numerical results for a loop-erased analogue of the RLRW.

The motivation for considering the random-length random walk model is easily understood.
In sufficiently high dimensions, it is known rigorously that, on $\mathbb{Z}^d$, the
Ising~\cite{Sakai2007}, SAW~\cite{Hara2008} and loop-erased random walk (LERW)\cite{LawlerLimic2010} two-point functions exhibit the same
scaling behaviour as the two-point function of a Simple Random Walk (SRW). 
Since the length of a SAW on the torus is necessarily finite, however, in order for SRW to accurately model SAW on the
torus it must be truncated to a finite length, denoted $\sN$. The resulting model is precisely the RLRW discussed
in~\cite{ZhouGrimmFangDengGaroni2018}. We note that in the special case in which $\sN$ is geometrically distributed, the two-point function
of RLRW on $\ZZ^d$ corresponds to the lattice Green function, which is very well studied; see~\cite{MichtaSladeLGF2021} and references
therein. In order to understand walk models on high dimensional tori, however, we will consider the case in which $\sN$ more closely mimics
the length of a corresponding SAW or Ising walk. 

This provides a motivation for studying the universal behaviour of the SAW and Ising walk length.
It has been proved that the expected walk length of critical SAW scales like the square root of the volume
both on the complete graph~\cite{Yadin2016}, and on the hypercube~\cite{Slade2021}. Universality would then suggest that the same behaviour
should hold for the critical SAW and Ising models on high-dimensional tori. While this remains an
open question, the analogous statement has recently been proved~\cite{MichtaSlade2021} for the Domb-Joyce model when $d>\dc$, provided the
interaction strength is sufficiently small.
Our simulations strongly suggest that the mean of the critical SAW and Ising walk lengths on high-dimensional tori do indeed scale as
$L^{d/2}$. Moreover, these simulations also suggest that the variance and standardised distribution function of the walk length of
the critical SAW and Ising models display the same universal behaviour known~\cite{DengGaroniGrimmNasrawiZhou2019,SladeKn2020} to hold for
SAW on the complete graph.

\subsection{Outline}
The outline of the remainder of this article is as follows. In Section~\ref{subsec:IsingWalks} we recall the definition of the Ising walk
introduced by Aizenman~\cite{Aizenman1982,Aizenman1985}, which holds on arbitrary graphs. Section~\ref{subsec:unwrapping and winding}
then provides a precise definition of the unwrapped two-point function for a general class of walk models defined on the discrete
torus. Section~\ref{subsec:SAW and Ising walk distributions} describes the specific SAW and Ising distributions that we consider on the
torus, and explains our method of simulating them. Section~\ref{subsec:RLRW definitions} recalls the relevant definitions for the RLRW and a
corresponding loop-erased analogue, while Section~\ref{subsec:numerical details} summarises the choices of parameters used in our simulations.
Section~\ref{sec:results} describes our results. Section~\ref{sec:walk length distribution} presents our numerical results for the SAW and
Ising walk lengths, and Section~\ref{subsec:proliferation of windings} presents numerical results for the number of
windings. Section~\ref{subsec:theorem} presents a general theorem on the two-point function of RLRW on $\ZZ^d$, and then utilises it to
predict the universal behaviour of the unwrapped two-point function of the SAW and Ising models on high-dimensional tori. These predictions
are then compared with the results from simulations. Section~\ref{sec:proofs} provides a proof for the proposition presented in
Section~\ref{subsec:theorem}. Finally, in the appendix we derive some identities for the two-point functions of RLRW and its loop-erased
analogue that were discussed in Section~\ref{subsec:RLRW definitions}.

\section{Models and observables} 
\label{sec:models and observables}

\subsection{Ising Walks}
\label{subsec:IsingWalks}
The zero-field ferromagnetic Ising model on finite graph $G=(V,E)$ at inverse temperature $\beta\ge0$ is defined by the measure
\begin{equation}
  \PP(\sigma) \propto \exp\left(\beta \sum_{ij\in E} \sigma_i\sigma_j\right), \qquad \sigma\in \{-1,1\}^V.
  \label{Ising measure}
\end{equation}
In this section, we briefly discuss a method due to Aizenman~\cite{Aizenman1982,Aizenman1985} for expressing the Ising two-point function
in terms of a particular random walk model.

We assume that $G$ is rooted, with root $0\in V$. For $v\in V\setminus 0$, let $\sC_v$ denote the set of all
$A\subseteq E$ such that the set of all vertices of odd degree in $(V,A)$ is precisely $\{0,v\}$, and let $\sC_0$ denote the set of all
$A\subseteq E$ such that $(V,A)$ has no vertices of odd degree. For a family of edge sets $S\subseteq 2^E$, let
\begin{equation}
  \lambda(S):=\sum_{A\in S} [\tanh(\beta)]^{|A|}.
\end{equation}
The high-temperature expansion for the Ising model (see e.g. \cite[(3.5)]{Aizenman1985} or~\cite[Lemma
  2.1]{CollevecchioGaroniHyndmanTokarev16}) implies that for all $v\in V$ we have
\begin{equation}
  \EE (\sigma_0\sigma_v) = \frac{\lambda(\sC_{v})}{\lambda(\sC_0)}.
  \label{eq:Ising high temperature}
\end{equation}
The expectation in~\eqref{eq:Ising high temperature} is with respect to the Ising measure~\eqref{Ising measure}.

Now, for $n\in\naturals$ let\footnote{Here, and in what follows, $\naturals$ denotes the set of non-negative integers, while $\posint$
  denotes the set of strictly positive integers.} $\Omega_G^n$ denote the set of all $n$-step walks on rooted graph $G=(V,E)$ which start at
the root $0$; i.e. all sequences $\omega_0,\ldots,\omega_n$ such that $\omega_i\in V$, $\omega_0=0$ and $\omega_i\,\omega_{i+1}\in E$.
We set $\Omega_G:=\bigcup_{n\in\naturals}\Omega_{G}^n$.
For $\omega\in \Omega_{G}^n$, the notation $\omega:0\to v$ implies $\omega_n=v$, and
we denote the end of $\omega$ by $e(\omega)=\omega_n$.
In all that follows, we let $|\omega|$ denote the number of steps, or \emph{length}, of the walk $\omega\in\Omega_G$, so that 
\begin{equation}
  |\omega|=n \text{ iff $\omega\in \Omega_{G}^n$}. \label{eq:walk length definition}
\end{equation}

Now fix an (arbitrary) ordering, $\prec$, of $V$. We define $\mathcal{T}:\cup_{v\in V}\sC_{v}\to\Omega_G$ as follows.
If $A\in \sC_0$, then $\sT(A)=0$. If $A \in \sC_{v}$ with $v\neq 0$, we recursively define
the walk $\sT(A) = v_0v_1...v_k$ from $v_0 = 0$ to $v_k = v$, such that from $v_i$ we choose $v_{i+1}$ to be the
smallest neighbour of $v_i$ such that $v_iv_{i+1} \in A$ and $v_i v_{i+1}$ has not previously been traversed by the walk. 
It is clear that $\mathcal{T}(A)$ defines an edge self-avoiding trail from $0$ to $v$.
An illustration of the construction is shown in~\Fref{fig:cartoon_ising_walk}.
\begin{figure}
\centering
\includegraphics[scale=0.4]{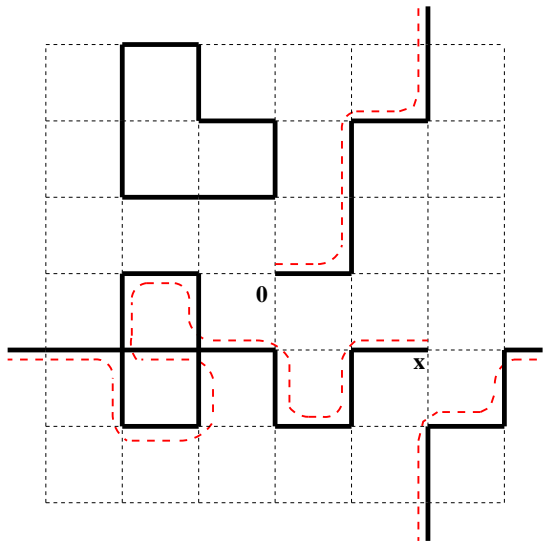}
\caption{Illustration of an Ising high-temperature graph and its corresponding Ising walk. The underlying graph, $G$, is
  the discrete two-dimensional torus with $L=7$, with the natural lexicographic order imposed on the vertices.
  The solid black lines denote a high-temperature edge configuration $A\in\sC_{x}$, while the red dashed line denotes $\sT(A)$. 
  The walk length is $|\sT|=23$.}
\label{fig:cartoon_ising_walk}
\end{figure}

Partitioning $\Omega_G$ in terms of $\sT$ we can write, for any $v\in V$,
\begin{align}
    \EE(\sigma_0\sigma_v) &= \sum_{\omega\in \Omega_G} \sum_{\substack{A\in\sC_{v} \\ \sT(A)=\omega}}
    \frac{[\tanh(\beta)]^{|A|}}{\lambda(\sC_0)}\\
    &= \sum_{\substack{\omega\in \Omega_G \\ \omega:0\to v}} \rho(\omega)\label{eq:Ising correlation via walk}
\end{align}
where $\rho:\Omega_G\to[0,\infty)$ is defined by 
\begin{equation}
  \rho(\omega):= \frac{\lambda(\sT^{-1}(\omega))}{\lambda(\sC_0)}.
  \label{eq:Ising two-point weight}
\end{equation}

Note that, by definition~\cite{MadrasSlade1996}, the two-point function for SAW on $G$ is again of the form \eqref{eq:Ising correlation via
  walk}, but with the weight given by 
\begin{equation}
  \rho(\omega) = J^{|\omega|}\ind(\text{$\omega$ is self-avoiding})
  \label{eq:SAW weight}
\end{equation}
where $J\in(0,\infty)$ is a parameter, referred to as the \emph{fugacity}.

\subsection{Unwrapping and winding}
\label{subsec:unwrapping and winding}
Let $\torus$ denote the $d$-dimensional discrete torus of period $L$. In what follows we identify the vertex set of $\torus$ with
$[-L/2,L/2)^d\cap\ZZ^d$. In defining $\Omega_{\torus}$ and $\Omega_{\ZZ^d}$ we take the root to be the origin.

  Let $\sW:\Omega_{\ZZ^d}\to \Omega_{\torus}$ denote the canonical bijection which wraps a $\ZZ^d$-walk onto a $\torus$-walk. Explicitly,
  for each $\omega\in\Omega_{\ZZ^d}$, the image $\tau=\sW(\omega)$ is defined recursively by setting $\tau_0=\omega_0=0$, and then for $0\le
  i \le |\omega|-1$ letting
  \begin{equation}
    \tau_{i+1}=\begin{cases}
    \tau_i + (\omega_{i+1}-\omega_i), & \tau_i+(\omega_{i+1}-\omega_{i}) \in \torus,\\
    \tau_i+(1-L)(\omega_{i+1}-\omega_{i}), &  \tau_i + (\omega_{i+1}-\omega_i) \not\in \torus.
    \end{cases}
    \label{eq:wrapping definition}
  \end{equation}

The number of windings of a $\torus$-walk along any specified coordinate axis can be conveniently expressed in terms of $\sW$. 
In particular, if $(x)_i$ denotes the $i$th coordinate of $x\in\ZZ^d$, the winding number of $\omega\in \Omega_{\torus}$ along the first
coordinate axis is
\begin{equation}
  \label{eq:winding number definition}
  \sR(\omega):=\left\lfloor\frac{|(e\circ\sW^{-1}(\omega))_1|}{L}\right\rfloor.
\end{equation}
As an illustration, consider $d=1$, $L=4$, and the walk $\omega=0,-1,-2,1,0,-1,-2,1$, which only takes steps to the left. Then
$|\omega|=7$, $e\circ\sW^{-1}(\omega)=-7$ and $\sR(\omega)=1$.

For given $\rho:\Omega_{\torus}\to[0,\infty)$ we define the corresponding two-point function $g_\rho: \torus\to[0,\infty)$ via
    \begin{equation}
      \label{eq:general two-point function definition}
       g_{\rho}(x) := \sum_{\tau\in \Omega_{\torus}} \rho(\tau) \ind[e(\tau)=x]
     \end{equation}
    and the corresponding \emph{unwrapped} two-point function $\tilde{g}_{\rho}:\ZZ^d\to [0,\infty)$ via
      \begin{align}
          \tilde{g}_\rho(z) &:= \sum_{\tau\in \Omega_{\torus}} \rho(\tau) \,\ind[e(\sW^{-1}(\tau)) = z ] \label{eq:unwrapped definition}\\
          &= \sum_{\zeta\in \Omega_{\ZZ^d}} \rho\circ\sW (\zeta) \,\ind[e(\zeta) = z] \label{eq:unwrapped alternative}.
      \end{align}
      We emphasise that if the weights are chosen via~\eqref{eq:Ising two-point weight} or~\eqref{eq:SAW weight}, then~\eqref{eq:general
        two-point function definition} reduces, respectively, to the Ising or SAW two-point functions considered in the previous section,
      specialised to the torus. Similarly, the unwrapped Ising and SAW two-point functions are defined by~\eqref{eq:unwrapped definition} specialised
      to~\eqref{eq:Ising two-point weight} and~\eqref{eq:SAW weight}, respectively.

      We also note that, following immediately from the definitions, we have
      \begin{equation}
        g_{\rho}(x) = \sum_{z\in \ZZ^d} \tilde{g}_{\rho}(x+zL).
      \end{equation}
      In this sense, the unwrapped two-point function is therefore a more fine-grained object than the torus two-point function.

      \subsection{SAW and Ising walk distributions}
      \label{subsec:SAW and Ising walk distributions}
We now describe in more detail the specific SAW and Ising walk ensembles which we study. We consider the variable-length ensemble of SAWs
on $\torus$, which corresponds to the set of all SAWs on $\torus$, rooted at the origin, and chosen randomly with a measure
proportional to the weight given in~\eqref{eq:SAW weight}. Let $\sS$ denote a random SAW chosen via this measure. We will be interested in
the distribution of the walk length $|\sS|$, defined in~\eqref{eq:walk length definition}, and winding number $\sR(\sS)$, defined
in~\eqref{eq:winding number definition}.  Moreover, it follows immediately from~\eqref{eq:unwrapped definition} and~\eqref{eq:SAW weight}
that the unwrapped SAW two-point function can be expressed in terms of $\sS$ via 
\begin{equation}
  \label{eq:SAW unwrapped two-point function estimator}
  \tilde{g}_{\rho}(z) = \frac{\PP[e(\sW^{-1}\circ\sS)=z]}{\PP[|\sS|=0]}, \qquad z\in\ZZ^d.
\end{equation}
Our simulations of $\sS$, discussed below, were performed using a lifted version~\cite{HuChenDeng2017} of the Berretti-Sokal
algorithm~\cite{BerrettiSokal1985}.

Now let us consider the Ising walk $\sT$. To begin, consider the probability measure on the state space
$\cup_{x\in\torus}\sC_{x}$, such that the probability of $A\in\cup_{x\in\torus}\sC_{x}$ is proportional to $[\tanh(\beta)]^{|A|}$. Let $\sA$
denote a random sample drawn from this measure. The distribution of $\sA$ is precisely the stationary distribution of the
Prokofiev-Svistunov worm algorithm~\cite{ProkofievSvistunov2001}, in which the worm tail is fixed to the origin.  Our simulations of $\sA$,
discussed below, were performed using such a worm algorithm. We will be interested in the induced distribution of $\sT(\sA)$. For
simplicity, we will henceforth adopt the abbreviation $\sT=\sT(\sA)$.

Analogously to SAW, it follows from~\eqref{eq:unwrapped definition} and~\eqref{eq:Ising two-point weight} that the
unwrapped Ising two-point function can be expressed exactly as in~\eqref{eq:SAW unwrapped two-point function estimator}, with $\sS$ replaced
by $\sT$.  Also analogously to SAW, we will again consider the induced distributions of $|\sT|$ and $\sR(\sT)$, which we refer to as the
Ising walk length and Ising winding number.

\subsection{Random-length random walks}
\label{subsec:RLRW definitions}
Let $(C_n)_{n \in \naturals}$ be an i.i.d.~sequence of uniformly random elements of $\{\pm e_1,...,\pm e_d\}$, where
$e_i=(0,\ldots,1,\ldots,0)\in\ZZ^d$ is the standard unit vector along the $i$th coordinate axis. Let $S_0=0$ and for $n\ge0$ set
$S_{n+1}=S_n+C_{n+1}$. Now let $\sN$ be an $\naturals$-valued random variable independent of $(C_n)_{n \in \naturals}$. The corresponding
\textit{Random-length Random Walk} on $\ZZ^d$ is the process $\sZ:=(S_n)_{n=0}^\sN$. Similarly, $\sX:=\sW(\sZ)$ is the corresponding RLRW on
$\torus$, where $\sW$ is the wrapping bijection defined in~\eqref{eq:wrapping definition}.

We also consider a loop-erased version of RLRW, constructed as follows. Recursively define a simple random walk $(R_i)_{i\in\NN}$ on $\torus$ by
applying~\eqref{eq:wrapping definition} to $(S_i)_{i\in\NN}$, and then perform chronological loop erasure on $(R_i)_{i\in\NN}$ until a walk
of length $\sN$ is generated. We refer to the resulting walk, denoted $\sL$, as the random-length loop-erased random walk (RLLERW) on
$\torus$. Note that $\sN$ must be bounded above by $L^d$ in order for $\sL$ to be well defined.

We define the two-point function of $\sZ$ to be
\begin{equation}
    \mathbb{E}\Bigg(\sum_{n=0}^{|\sZ|} \ind(\sZ_n= x)\Bigg)
\label{def:greens_function_rlrw}
\end{equation}
which gives the expected number of visits of $\sZ$ to $x \in \mathbb{Z}^d$. Analogous definitions hold for the RLRW and RLLERW on
$\torus$ by replacing $\sZ$ with $\sX$ and $\sL$, respectively.
As noted in the Introduction, in the special case in which $\sN$ is geometrically distributed, the two-point function of RLRW on $\ZZ^d$
corresponds to the lattice Green function, which is very well studied; see~\cite{MichtaSladeLGF2021} and references therein.

A simple rearrangement of~\eqref{def:greens_function_rlrw} (see~\ref{sec:Appendix A}) shows that it can be expressed in the
form~\eqref{eq:Ising correlation via walk} with
\begin{equation}
  \label{eq:RLRW weight}
  \rho(\omega) = \frac{\PP(\sN\ge|\omega|)}{(2d)^{|\omega|}}, \qquad \omega\in\Omega_{\ZZ^d}.
\end{equation}
Precisely the same statement also holds for $\sX$, with the same weights, but replacing $\Omega_{\ZZ^d}$ with $\Omega_{\torus}$.
Moreover, an analogous statement also holds for $\sL$ with (see~\ref{sec:Appendix A})
\begin{equation}
  \label{eq:RLLERW weight}
  \rho(\omega) = \PP(\sL\sqsupseteq \omega),
\end{equation}
where for any two walks $\tau,\omega\in\Omega_{\torus}$, the notation $\tau\sqsupseteq\omega$ implies that $|\tau|\ge|\omega|$ and
$\tau_i=\omega_i$ for all $0\le i \le |\omega|$. 

The unwrapped two-point functions of $\sZ$, $\sX$, and $\sL$ are defined by~\eqref{eq:unwrapped definition}, with the appropriate choices of
weight $\rho$ just outlined. Now, since for any $\omega\in \Omega_{\ZZ^d}$ we have $|\sW(\omega)|=|\omega|$, it follows
from~\eqref{eq:unwrapped alternative} and~\eqref{eq:RLRW weight} that the unwrapped two-point function of $\sX$ is simply 
\begin{align}
  \tilde{g}_{\rho}(z) &= \sum_{\zeta\in \Omega_{\ZZ^d}}\PP(\sN\ge|\zeta|) \frac{\ind(e(\zeta)=z)}{(2d)^{|\zeta|}} \nonumber\\
    &=\sum_{n=0}^\infty\PP(\sN\ge n) \sum_{\zeta\in \Omega_{\ZZ^d}^n} \frac{\ind(e(\zeta)=z)}{(2d)^{|\zeta|}} \nonumber\\
    &=\sum_{n=0}^\infty\PP(\sN\ge n) \PP(S_n=z) \nonumber\\
  &=\EE \left(\sum_{n=0}^{|\sZ|} \ind(\sZ_n=z)\right).
  \label{eq:unwrapped RLRW g on torus is standard g of RLRW on lattice}
\end{align}
In other words, the unwrapped two-point function of the RLRW on the torus is simply the two-point function of the corresponding RLRW on
$\ZZ^d$. Now, for an appropriate choice of distribution for $\sN$, the unwrapped two-point function of $\sX$
is expected to display the same asymptotics as the unwrapped two-point functions for the SAW and Ising walk.
This then motivates studying the two-point function of $\sZ$, which we do in Sections~\ref{subsec:theorem} and~\ref{sec:proofs}.

Finally, we note that, after some rearrangement (see~\ref{sec:Appendix A}), the unwrapped two-point function of $\sL$ can be
expressed as 
\begin{equation}
  \label{eq:LERW unwrapped two-point function}
 \tilde{g}_{\rho}(z) = \PP[\sW^{-1}(\sL)\ni z],
\end{equation}
which can be easily estimated via simulation. 

\subsection{Numerical details}
\label{subsec:numerical details}
Our simulations of the Ising model were performed at the exact infinite-volume critical point in two dimensions~\cite{Baxter2016}, and at
the estimated location of the infinite-volume critical point $\tanh(\beta_{\mathrm{c}}) = 0.113~424~8(5)$~\cite{LundowMarkstrom2014} in five 
dimensions. The SAW model was simulated at the estimated location of the infinite-volume critical points,
$J_{\mathrm{c}}=0.379~052~277~758(4)$~\cite{Jensen2003} in two dimensions, $J_{\mathrm{c}} = 0.113~140~84(1)$~\cite{HuChenDeng2017} in five
dimensions, and $J_{\mathrm{c}}=0.091~927~86(4)$~\cite{Owczarek2001} in six dimensions.

For the Ising model, we simulated linear system sizes up to $L=31$ in five dimensions. For SAW, we simulated linear system sizes up to
$L=221$ in five dimensions, and $L=57$ in six dimensions. For the RLLERW, we simulated linear system sizes up to $L=161$ in five dimensions.

Our error estimation follows standard procedures, see for instance~\cite{Young2015,Sokal1996}. Analyses of integrated autocorrelation times
for the worm and irreversible Berretti-Sokal algorithms are presented in~\cite{DengGaroniSokal2007} and~\cite{HuChenDeng2017},
respectively.

\section{Results}
\label{sec:results}

\subsection{Universal walk length distribution}
\label{sec:walk length distribution}
Let $\sK$ denote a self-avoiding walk on the complete graph $K_n$, rooted at a fixed vertex, distributed according to the variable-length
ensemble. The probability distribution of $\sK$ is then proportional to the weight given in~\eqref{eq:SAW weight}.
It was shown in~\cite{Yadin2016} that the critical fugacity for $\sK$ occurs at $J=1/n$. Furthermore, at criticality, it is
known~\cite[Theorem 1.1]{DengGaroniGrimmNasrawiZhou2019} (see also~\cite{Yadin2016,SladeKn2020}) that
\begin{equation}
  \begin{split}
  \EE(|\sK|) &\sim \sqrt{\frac{2}{\pi}} \sqrt{n} \\
  \var(|\sK|) &\sim \left(1-\frac{2}{\pi}\right) n.
  \label{eq:Kn SAW cumulants}
  \end{split}
\end{equation}
From universality, one would then expect that if one considered the walk lengths of the critical SAW or Ising models on $\torus$ with
$d>\dc$, then their means should scale as $L^{d/2}$, and their variances should scale as $L^d$.
\Fref{fig:walk length mean and variance on torus} provides strong evidence that this is the case.

As a simple consequence, this would imply that the ratio of the mean and standard deviation of the walk length therefore converges to a
positive constant. From~\eqref{eq:Kn SAW cumulants}, the value of this constant for $\sK$ is
\begin{equation}
  \lim_{n\to\infty} \frac{\EE(|\sK|)}{\sqrt{\var(|\sK|)}} = \sqrt{\frac{2}{\pi-2}} =: \varphi
  \label{eq:Kn SAW moment ratio}
\end{equation}
For comparison, the analogous ratio for the SAW and Ising models on $\torus$ is plotted in \Fref{subfig:walk length moment ratio}.
The SAW data suggest it is plausible, for both $d=5$ and $d=6$, that $\EE(|\sS|)/\sqrt{\var(|\sS|)}$ is converging to
the complete graph value, $\varphi$. The $d=5$ Ising data suggest, however, that $\EE(|\sT|)/\sqrt{\var(|\sT|)}$ is converging to a constant
strictly less than $\varphi$, although it is certainly numerically close to $\varphi$.

In addition to the asymptotic moments given in~\eqref{eq:Kn SAW cumulants}, central limit theorems have been established for $\sK$. Indeed, it
follows from~\cite[Theorem 1.2]{DengGaroniGrimmNasrawiZhou2019} (see also \cite[Theorem 1.3]{SladeKn2020}) that, at criticality, 
\begin{equation}
\lim_{n\to\infty} \PP\left(\frac{|\sK| - \EE(|\sK|)}{\sqrt{\var(|\sK|)}} \le x\right) = \PP\left(\frac{|X|-\EE(|X|)}{\sqrt{\var(|X|)}}\le x\right)
\label{eq:standardised limit theorem for Kn SAW}
\end{equation}
for all $x\in\RR$, where $X$ is a standard normal random variable.
We note that the law of $|X|$ is the half-normal distribution, which can be given explicitly by
\begin{equation}
  \PP(|X|\le x) = \ind(x>0)[1-2\Phibar(x)]
  \label{eq:half-normal distribtion function}
\end{equation}
where $\Phibar$ denotes the standard normal tail distribution, so that for all $x\in\RR$ 
\begin{equation}
\Phibar(x) := \frac{1}{\sqrt{2\pi}}\int_{x}^\infty \e^{-s^2/2}\diff s.
\label{eq:normal tail distribtion}
\end{equation}

For later reference, we shall denote by $F$ the law of the standardised version of $|X|$ appearing on the right-hand side
of~\eqref{eq:standardised limit theorem for Kn SAW}, i.e. for $x\in\RR$
\begin{equation}
  F(x):= \PP\left(\frac{|X|-\EE(|X|)}{\sqrt{\var(|X|)}}\le x\right).
  \label{eq:def of standardised half normal}
\end{equation}
By universality, one would expect that the standardised distribution functions of $|\sS|$ and $|\sT|$ on $\torus$ should also converge to $F$.
\Fref{fig:walk length distribution on torus} (right panel) provides strong evidence that this is indeed the case.

\begin{figure}
  \centering
   \begin{subfigure}{0.425\textwidth}
    \centering
    \includegraphics[width=\textwidth]{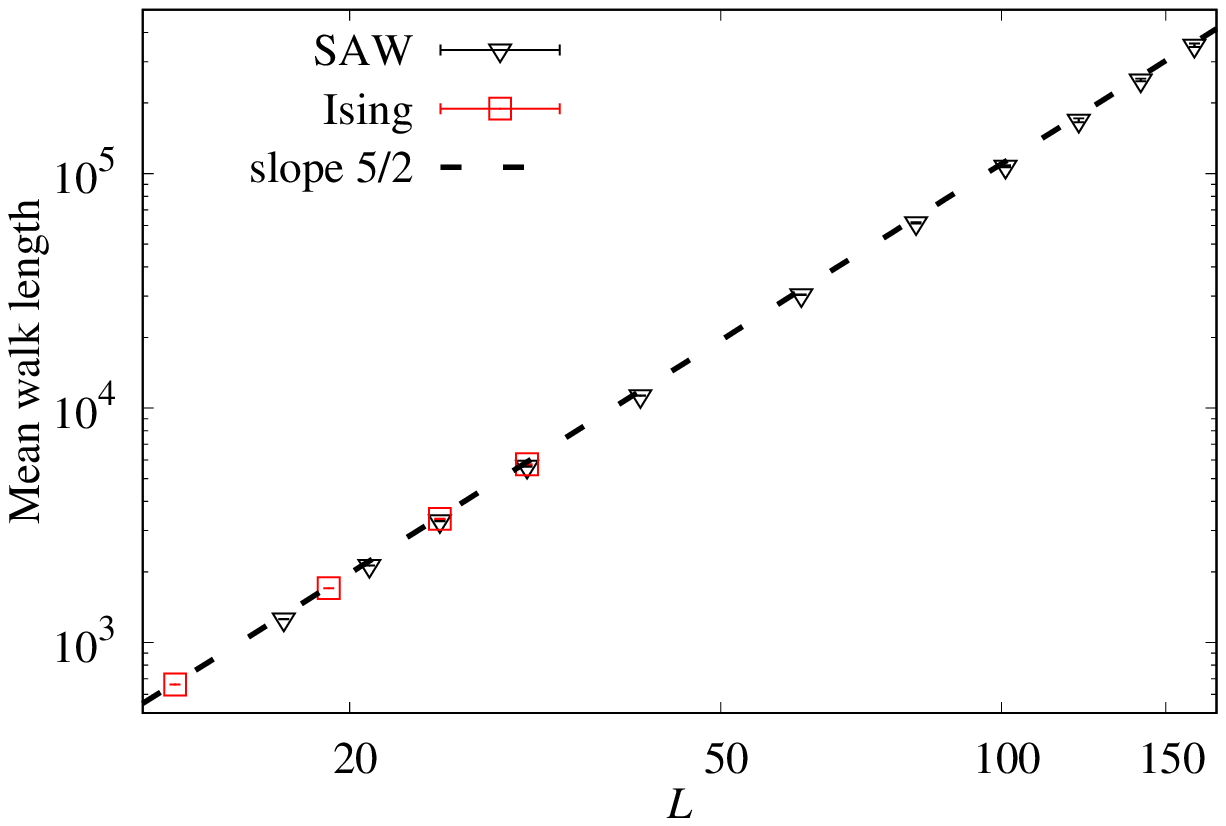}
    \caption{\label{subfig:walk length mean}}
   \end{subfigure}
   \begin{subfigure}{0.425\textwidth}
    \centering
    \includegraphics[width=\textwidth]{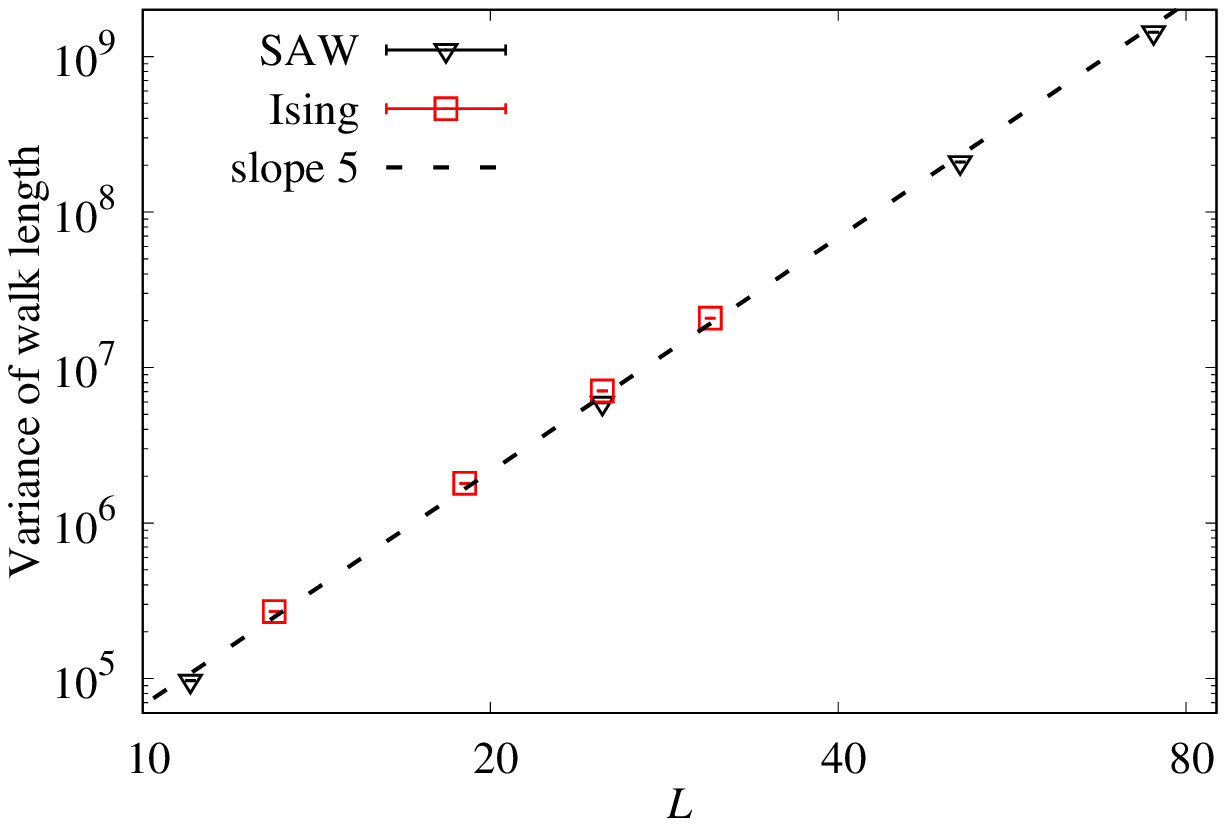}
    \caption{\label{subfig:walk length variance}}
   \end{subfigure}
   \caption{
     (\subref{subfig:walk length mean}) Simulated mean of the critical SAW and Ising walk lengths on five-dimensional tori. 
     (\subref{subfig:walk length variance}) Simulated variance of the critical SAW and Ising walk lengths on five-dimensional tori.
   }
\label{fig:walk length mean and variance on torus}
\end{figure}

\begin{figure}
  \centering
     \begin{subfigure}{0.455\textwidth}
    \centering
    \includegraphics[width=\textwidth]{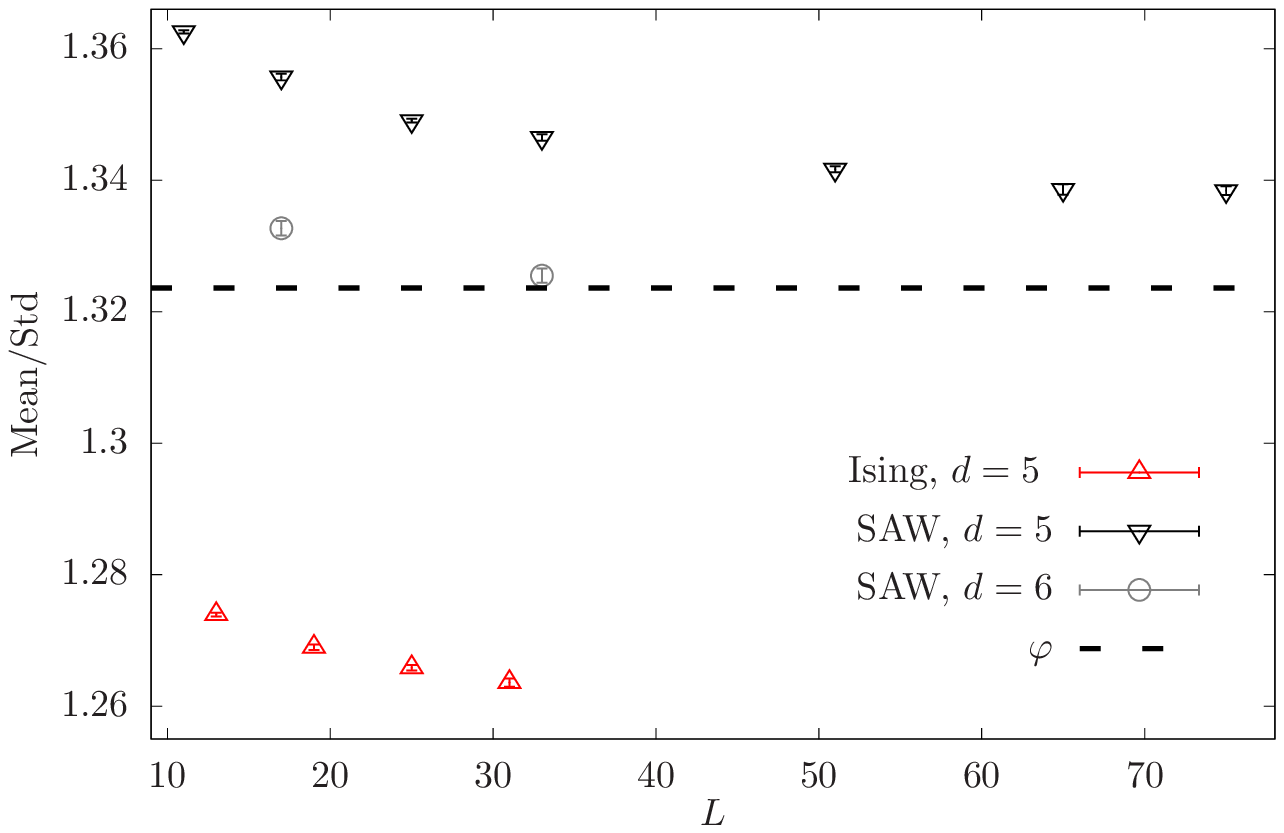}
    \caption{\label{subfig:walk length moment ratio}}
   \end{subfigure}
     \begin{subfigure}{0.455\textwidth}
    \centering
    \includegraphics[width=\textwidth]{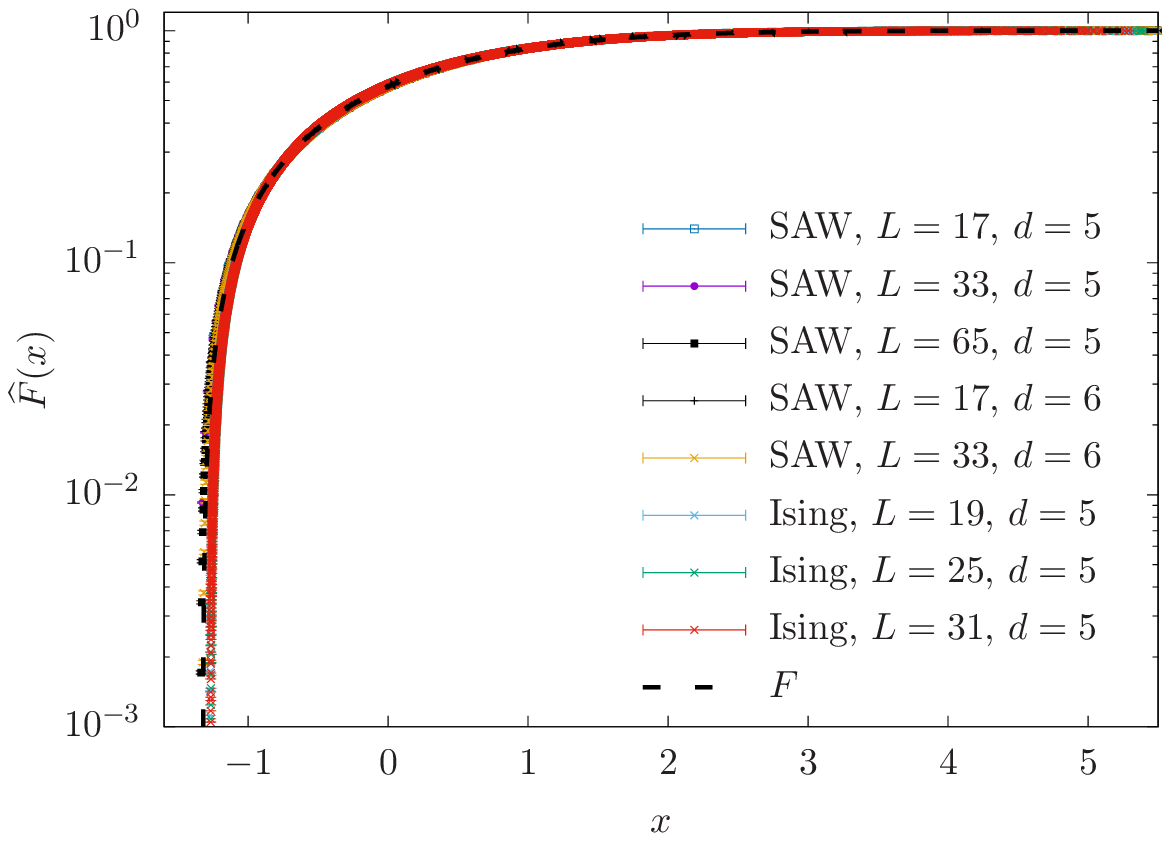}
    \caption{\label{subfig:walk length histograms}}
     \end{subfigure}
     \caption{
       (\subref{subfig:walk length moment ratio}) Ratio of simulated mean and standard deviation of the critical SAW walk length on five- and
       six-dimensional tori, and the critical Ising walk length on five-dimensional tori.
       The dashed curve corresponds to the limiting value $\varphi$ for the case of SAW on $K_n$; see~\eqref{eq:Kn SAW moment ratio}.
       (\subref{subfig:walk length histograms})
       Simulated distribution function, $\widehat{F}$, of the standardised walk length $(|\sS| - \EE(|\sS|))/\sqrt{\var(|\sS|)}$
       for critical SAW on five- and six-dimensional tori, as well as the standardised critical Ising walk length
       $(|\sT| - \EE(|\sT|))/\sqrt{\var(|\sT|)}$ on five-dimensional tori.
       The dashed curve corresponds to $F$ given in~\eqref{eq:def of standardised half normal}.
     }
     \label{fig:walk length distribution on torus}
\end{figure}

\subsection{Proliferation of windings}
\label{subsec:proliferation of windings}
We now consider the large $L$ asymptotics of $\EE(\sR)$. Fig.~\ref{fig:winding} plots $\EE(\sR)$ with $d=2,5,6$ for SAW, and $d=2,5$ for the
Ising model. In dimensions below $\dc$, we find that $\EE(\sR)$ is bounded as $L \to \infty$. By contrast, we observe that windings
proliferate for $d>\dc$. It was conjectured in~\cite{GrimmElciZhouGaroniDeng2017} that $\EE(\sR)$ should scale as $L^{d/4 -1}$ at
criticality when $d>\dc$. For $d=5$, fitting $\EE(\sR)$ to a power law ansatz produces an exponent value of $0.24(3)$ for the Ising
model and $0.30(6)$ for SAW. For $d=6$ SAW, the analogous fit yields an exponent value of $0.46(6)$. In each case, the estimated and conjectured
exponent values agree within error bars. We note that, in the Ising case, the definition of $\sR$ considered here differs from that
used in~\cite{GrimmElciZhouGaroniDeng2017}, the current version be a more natural analogue of the SAW definition. The asymptotic behaviour is
the same in both cases however.

Finally, we also studied the average winding number of a RLLERW with $d=5$ whose walk length is drawn from the asymptotic walk length distribution of
the complete-graph SAW; i.e. with standardised distribution function $F$, and with mean and variance
given by the right-hand side of~\eqref{eq:Kn SAW cumulants} with $n=L^d$. 
Our fits lead to the exponent value $0.29(5)$, in agreement with the SAW and Ising models.

\begin{figure}
\centering
\includegraphics[scale=0.7]{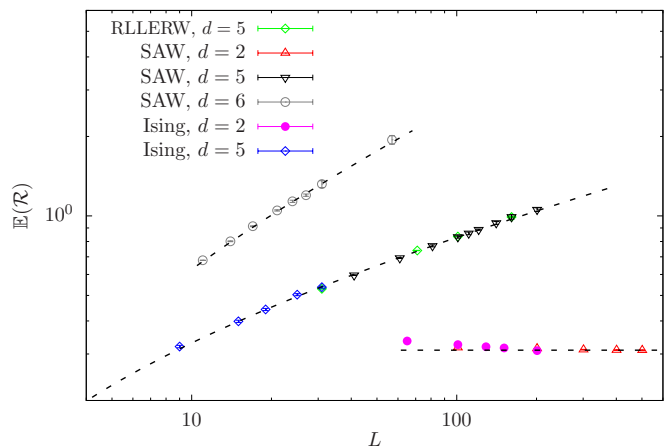}
\caption{Simulated values of $\EE(\sR)$ for the critical Ising and SAW models, on discrete tori in various dimensions.
Analogous results are also shown for RLLERW with $d=5$ and walk length chosen
via the asymptotic distribution of SAW on the complete-graph.
To emphasise the universal scaling, the data for all models in each given dimension were translated onto a single curve by multiplying by
suitable ($L$-independent) constants. The number of windings is clearly asymptotically constant in $L$ for $d<\dc$, while above $\dc$ windings
proliferate as $L$ increases.}
\label{fig:winding}
\end{figure}

\subsection{Unwrapped two-point functions}
\label{subsec:theorem}
We begin by stating the following proposition for the two-point function of RLRW on $\mathbb{Z}^d$. The proof is deferred to
\Sref{sec:proofs}. We emphasise that, due to \eqref{eq:unwrapped RLRW g on torus is standard g of RLRW on lattice},
Proposition~\ref{prop:infinite_lattice} also immediately implies the analogous result for the unwrapped two-point function of RLRW on the torus.
\begin{proposition}
  \label{prop:infinite_lattice}
  Consider a sequence of $\naturals$-valued random variables $\sN_L$, such that there exists a non-decreasing sequence $a_L>0$ for which $\sN_L/a_L$
  converges in distribution, as $L\to\infty$, to a random variable with distribution function $G$. Now fix an integer $d\ge3$, and let
  $z_L\in\ZZ^d$ be a sequence such that $\|z_L\|\to\infty$ as $L\to\infty$, with 
  $\xi:=\lim_{L\to\infty}\|z_L\|/\sqrt{a_L}\in(0,+\infty]$ well defined~\footnote{As an element of the extended reals; as
  $L\to\infty$, either $\|z_L\|^2/a_L$ converges, or it diverges to $+\infty$.}. Then, the two-point function of $\sZ$ satisfies 
  \begin{equation*}
    \lim_{L\rightarrow \infty} \|z_L\|^{d-2} g(z_L)  =
    \frac{d}{2\pi^{d/2}}\int_{0}^\infty s^{d/2 - 2}\e^{-s}\,\left[1-G\left(\frac{d}{2}\frac{\xi^2}{s}\right)\right] \diff s
  \end{equation*}
\end{proposition}

As a first observation, we note that, provided $G$ is continuous at the origin, as $\xi\to0$ the right-hand side of the limit appearing in Proposition~~\ref{prop:infinite_lattice} reduces to 
\begin{equation}
  \lim_{\xi\to0}  \frac{d}{2\pi^{d/2}}
  \int_{0}^\infty s^{d/2 - 2}\e^{-s}\,\left[1-G\left(\frac{d}{2}\frac{\xi^2}{s}\right)\right] \diff s = 
  \frac{d}{2\pi^{d/2}} \Gamma(d/2-1)
\end{equation}
in agreement with the well-known asymptotics of the two-point function of simple random walk (see e.g.~\cite[Theorem
  4.3.1]{LawlerLimic2010}). This is to be expected, since typical walks of length $a_L$ explore a ball whose radius is of order
$\sqrt{a_L}$, and $\xi\to0$ corresponds to the case where $\sqrt{a_L}$ dominates the spatial scale $\|z_L\|$ probed, meaning the walk length
grows so fast that the finiteness of the walk is not observed.

We are particularly interested in the case where $G$ corresponds to the SAW and Ising models on high-dimensional tori. 
The numerical results of \Sref{sec:walk length distribution} lead to the conjecture that for $d>4$ at criticality
$\EE(|\sS|) \sim B_{\sS,d} L^{d/2}$ and $\sqrt{\var(|\sS|)}\sim A_{\sS,d} L^{d/2}$, and $(|\sS|-\EE(|\sS|))/\sqrt{\var(|\sS|)}$ converges weakly
to $F$, where $A_{\sS,d},B_{\sS,d}>0$ and $F$ is as given in~\eqref{eq:def of standardised half normal}. Assuming the
validity of this conjecture, it follows from standard convergence of types arguments (see e.g.~\cite[pp. 193]{Billingsley94}) that
\begin{equation}
  \lim_{L\to\infty}\PP\left(\frac{\sS}{L^{d/2}}\le x\right)= F\left(\frac{x-B_{\sS,d}}{A_{\sS,d}}\right)
  \label{eq:conjectured limit of SAW distribution function}
\end{equation}
for all $x\in\RR$.

Now let $\sN=|\sS|$, $a_L=L^{d/2}$ and for fixed $\xi\in(0,\infty)$ let $z_L=\lfloor L^{d/4}\xi\rfloor e_1$. Assuming the validity
of~\eqref{eq:conjectured limit of SAW distribution function} it follows from Proposition~\ref{prop:infinite_lattice} that as $L\to\infty$
the unwrapped two-point function of the corresponding RLRW on $\torus$ satisfies
\begin{equation}
  \|z_L\|^{d-2}\,\tilde{g}(z_L) \sim
  H_d(1,1/A_{\sS,d},B_{\sS,d}/A_{\sS,d};\xi)
  \label{eq:unwrapped RLRW g with SAW walk length}
\end{equation}
where
\begin{equation}
H_d(\alpha,\beta,\gamma;\xi):= \alpha \frac{d}{2\pi^{d/2}}\,\int_{0}^{\infty} s^{d/2-2} e^{-s}
\left[1-F\left(\beta\,\frac{d\xi^2}{2s} -\gamma\right)\right]\diff s.
\end{equation}

Universality then makes it natural to conjecture that the asymptotics of $\|z_L\|^{d-2}\,\tilde{g}(z_L)$ for the SAW and Ising models on
the torus should also be given by $H_d(\alpha,\beta,\gamma;\xi)$, for suitable model-dependent values of the constants, $\alpha,\beta,\gamma$. 
Figures~\ref{subfig:unwrappedSAW} and \ref{subfig:unwrappedIsing} provide strong evidence in favour of these conjectures.
In Figure~\ref{subfig:unwrappedSAW}, the constants for SAW are set to $\alpha=0.85$, $\beta=1.5/A_{\sS,d}$, and $\gamma=
B_{\sS,d}/A_{\sS,d}$, while in~\ref{subfig:unwrappedIsing} the constants for the Ising model are set to $\alpha = 1$, $\beta=1.2/A_{\sT,d}$ and
$\gamma=B_{\sT,d}/A_{\sT,d}$. 

In addition to the Ising and SAW cases, in \Fref{subfig:unwrappedRLLERW} we plot the two-point function for RLLERW with
walk length chosen via the asymptotic distribution of SAW on the complete-graph,
which again appears to be described by $H_d(\alpha,\beta,\gamma;\xi)$.
In this case, we set $\alpha = 0.75$, $\beta=1.2/\sqrt{1-2/\pi}$ and $\gamma=\varphi$; c.f.~\eqref{eq:Kn SAW cumulants} and~\eqref{eq:Kn SAW
  moment ratio}.

We note that the two-point functions of the critical SAW~\cite[Theorem 1.1]{Hara2008} and Ising~\cite[Theorem 1.3]{Sakai2007} models on
$\ZZ^d$ are known to satisfy
\begin{equation}
  \lim_{\|z\|\to \infty} \|z\|^{d-2}\,g(z) = A \frac{d}{2\pi^{d/2}}\Gamma(d/2-1)
  \label{eq:Hara asymptotics}
\end{equation}
where the non-universal constant $A$ can be expressed in terms of quantities appearing in the lace expansion.
Our conjecture, if true, would therefore provide natural
finite-size analogues/refinements of~\cite[Theorem 1.1]{Hara2008} and~\cite[Theorem 1.3]{Sakai2007}.
We remark that for SAW in $d=5$ it follows rigorously from bounds\footnote{Specifically, Equations (1.21), (1.25) and the connective
  constant bound on page 238} established in~\cite{HaraSladeRMP1992} that $0.81<A<0.92$;
the value of $\alpha=0.85$ used in Figure~\ref{subfig:unwrappedSAW} for $d=5$ SAW is therefore consistent with these rigorous bounds on the
value of $A$.

\begin{figure*}[ht!]
  \centering
   \begin{subfigure}{0.425\textwidth}
    \centering
    \includegraphics[width=\textwidth]{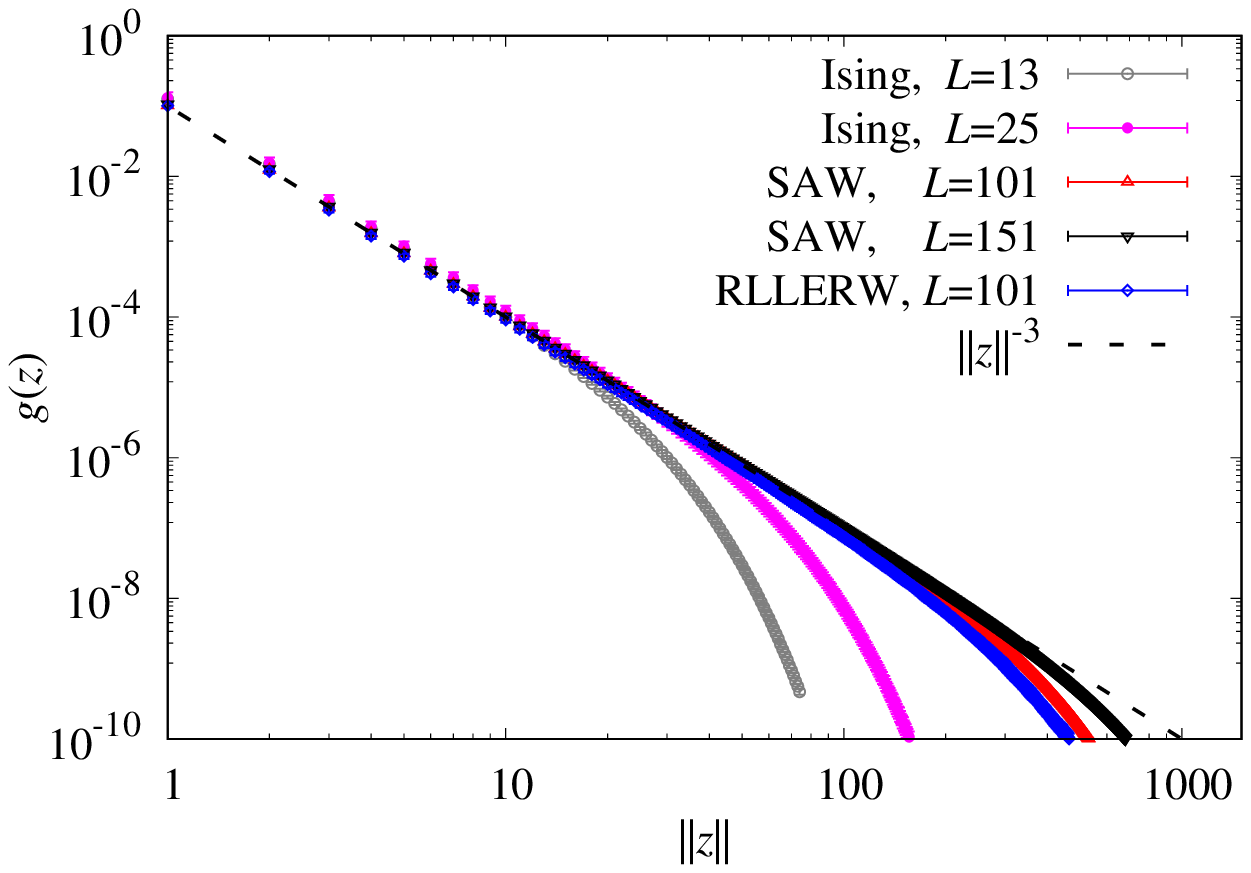}
    \caption{\label{subfig:unwrapped}}
   \end{subfigure}
   \begin{subfigure}{0.425\textwidth}
    \centering
    \includegraphics[width=\textwidth]{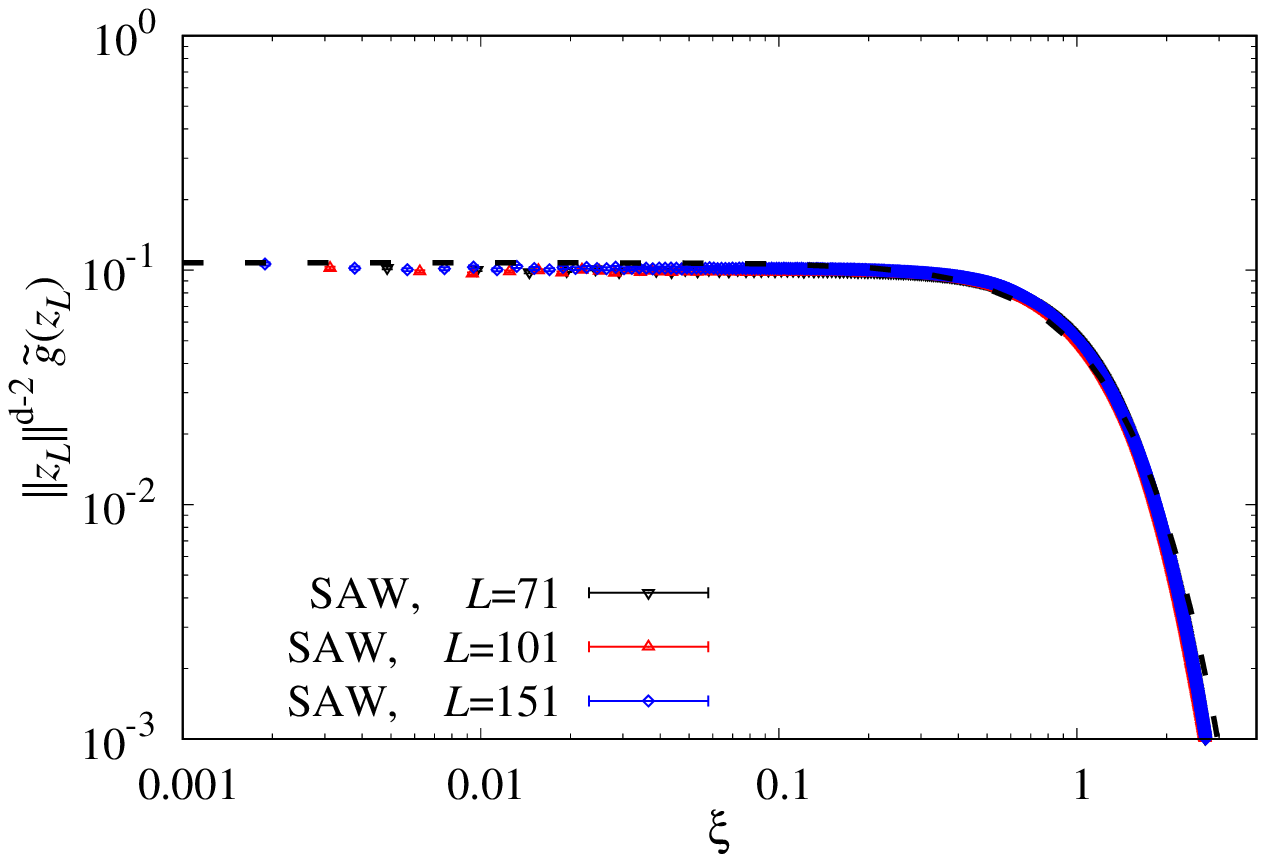}
    \caption{\label{subfig:unwrappedSAW}}
  \end{subfigure}
   \begin{subfigure}{0.425\textwidth}
    \centering
    \includegraphics[width=\textwidth]{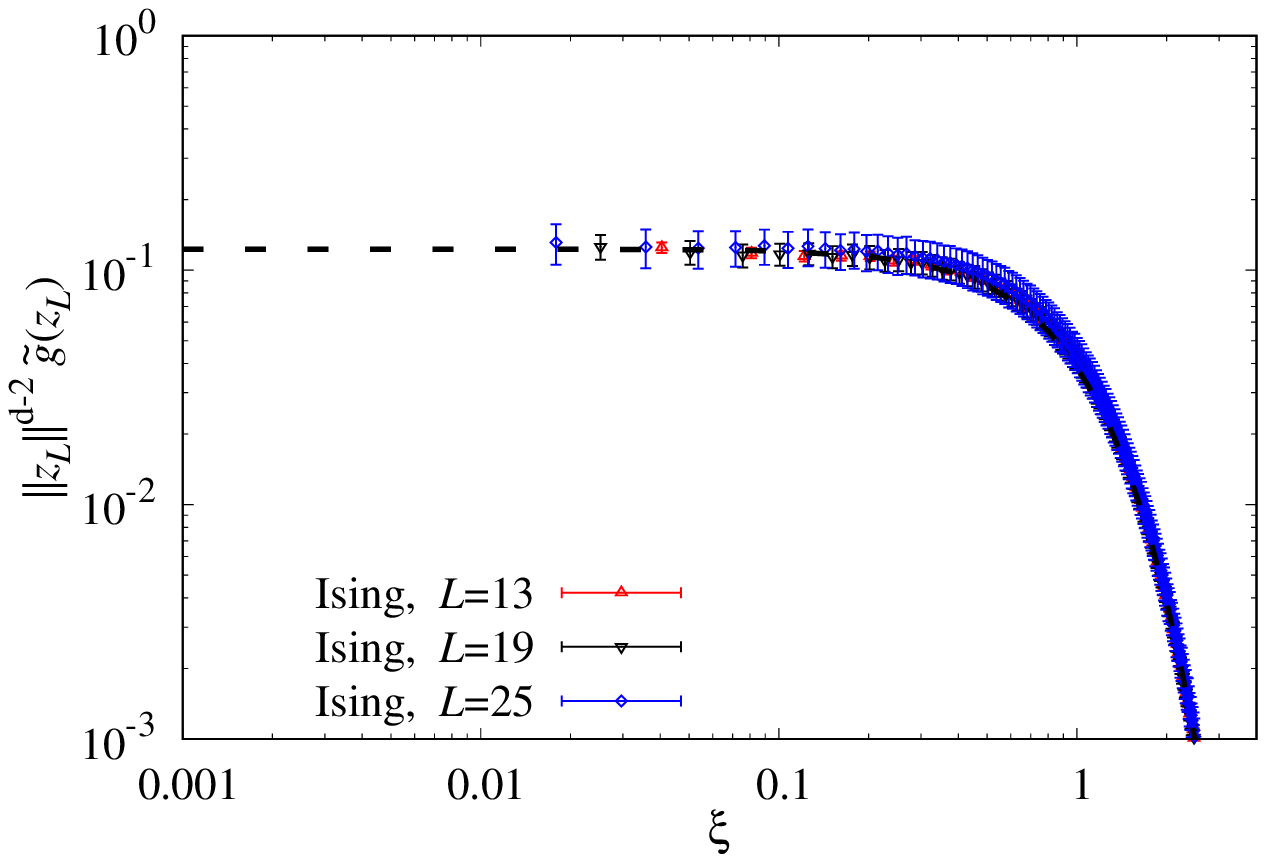}
    \caption{\label{subfig:unwrappedIsing}}
  \end{subfigure}
   \begin{subfigure}{0.425\textwidth}
    \centering
    \includegraphics[width=\textwidth]{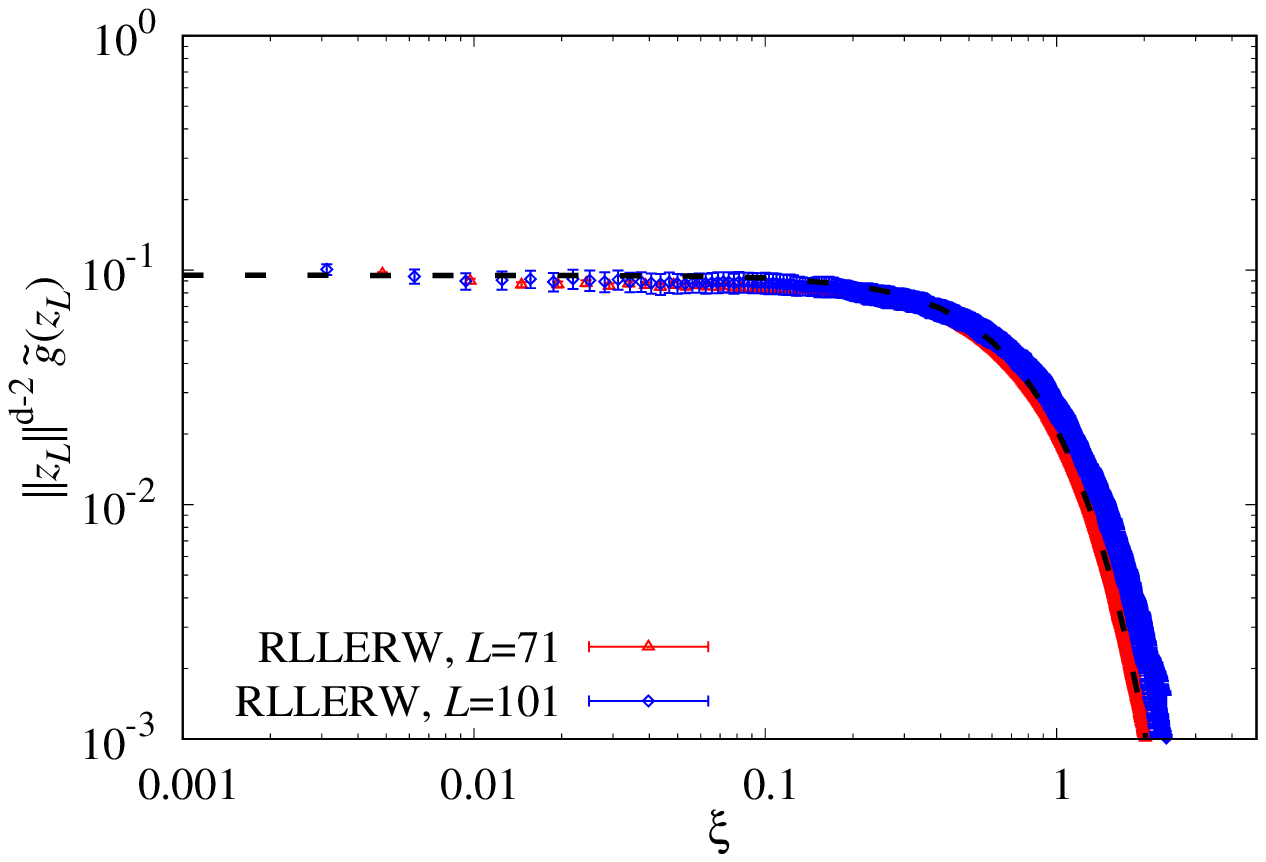}
    \caption{\label{subfig:unwrappedRLLERW}}
  \end{subfigure}
\caption{(\subref{subfig:unwrapped}) Unwrapped two-point functions on the five-dimensional torus, of the critical Ising and SAW models, and
  RLLERW whose walk length is drawn from the asymptotic walk length distribution of the complete-graph SAW.
  Standard SRW behaviour is clearly displayed in the bulk of the system.
  (\subref{subfig:unwrappedSAW})
  Plot of $\|z_L\|^{d-2}\tilde{g}(z_L)$ vs $\xi$ for SAW on five-dimensional tori. The dashed curve shows $H(\alpha,\beta,\gamma;\xi)$
  with constants $\alpha,\beta,\gamma$ set to the values described in the text, with $A_{\sS,d}$ and $B_{\sS,d}$ estimated via simulation.
  (\subref{subfig:unwrappedIsing}) Analogous plot to (\subref{subfig:unwrappedSAW}), for the Ising case.
  (\subref{subfig:unwrappedRLLERW}) Analogous plot to (\subref{subfig:unwrappedSAW}), for case of RLLERW whose walk length is drawn from the
  asymptotic walk length distribution of the complete-graph SAW. 
}
\label{fig:unwrapped_two-point_function}
\end{figure*}

\section{Proof of Proposition~\ref{prop:infinite_lattice}}
\label{sec:proofs}
Let $(S_n)_{n=0}^{\infty}$ be a simple random walk on $\ZZ^d$, starting from the origin, and let
\begin{equation}
p_n(z):=\mathbb{P}(S_n=z), \qquad z\in\ZZ^d.
\end{equation}
We say that $n\in\naturals$ and $z\in\ZZ^d$ have the same parity, and write $n\leftrightarrow
z$, iff $n + \|z\|_1$ is even. Clearly, $p_n(z) = 0$ if $n\nleftrightarrow z$. The main tool used to prove 
Proposition~\ref{prop:infinite_lattice} is the local central limit theorem for $(S_n)_{n=0}^{\infty}$, which allows $p_n(z)$ to be
approximated, when $n$ is large, by 
\begin{equation}
\bar{p}_n(z):= 2\left(\frac{d}{2\pi n}\right)^{d/2} \exp\left(-\frac{d\|z\|^2}{2n}\right), \qquad z\in\ZZ^d, \,n\ge1.
\label{Eq: pnbar(x)}
\end{equation}
In particular, we will apply the following lemma, whose proof we defer until the end of this section.

\begin{lemma}
  \label{lem:p and pbar sum bounds}
  Fix a positive integer $d$, and let $z_L\in\ZZ^d$ be a sequence for which $\|z_L\|\to\infty$ as $L\to\infty$.
  Then for any $\epsilon>0$, as $L\to \infty$
  \begin{enumerate}
  \item\label{lem_part:LCLT} $\sum\limits_{n=1}^\infty |p_n(z_L) - \bar{p}_{n}(z_L)| = O(\|z_L\|^{-d+\epsilon})$
  \item\label{lem_part:distance between consecutive pbar}
    $\sum\limits_{n=1}^\infty |\bar{p}_n(z_L) - \bar{p}_{n+1}(z_L)| = O(\|z_L\|^{-d+\epsilon})$
    \end{enumerate}
\end{lemma}

\begin{proof}[Proof of Proposition~\ref{prop:infinite_lattice}]
  Let $\sN$ be an $\naturals$-valued random variable, independent of $(S_n)_{n=0}^{\infty}$. It follows from the
  definition~\eqref{def:greens_function_rlrw} that for all $z\in\ZZ^d$
  $$
  g(z) = \EE\, \sum_{n=0}^{\infty} \ind(\sN\ge n)\ind(S_n=z) = \sum_{n=0}^{\infty}\, \PP(\sN\ge n)\, p_n(z).
  $$
  Moreover, if $z\neq0$ we have
  \begin{align*}
    g(z) &= \sum_{n=1}^{\infty}\, \PP(\sN\ge n)\, p_n(z) \,\ind(z \leftrightarrow n) \\
    &= D(z) + E_1(z) + E_2(z)
  \end{align*}
  where
  \begin{align}
    D(z) &:= \sum_{n=1}^{\infty} \frac{\bar{p}_n(z)}{2} \,\PP(\sN\ge n), \label{eq:D definition}\\
    E_1(z) &:=
    \sum_{n=1}^{\infty} \frac{\bar{p}_n(z)}{2} \,\PP(\sN\ge n)\ind(z\leftrightarrow n)
  - \sum_{n=1}^{\infty} \frac{\bar{p}_n(z)}{2} \,\PP(\sN\ge n)\ind(z \not\leftrightarrow n),
    \label{eq:E1 definition}\\
    E_2(z) &= \sum_{n=1}^{\infty}\, \PP(\sN\ge n)\, \left[p_n(z)-\bar{p}_n(z)\right] \,\ind(z \leftrightarrow n)\label{eq:E2 definition}.
  \end{align}
  We consider each of these three terms in turn, beginning with $D$.
  If $a:(0,\infty)\to(0,\infty)$ is non-increasing and $b:(0,\infty)\to(0,\infty)$ is non-decreasing, then for any positive integer $k$ one has
  \begin{equation}
    \label{eq:bounding sums via integrals}
    \int_{k-1}^{\infty} a(t+1)b(t)\diff t \le \sum_{n=k}^{\infty} a(n)b(n) \le \int_{k}^{\infty}a(t-1)b(t)\diff t.
  \end{equation}
  Applying~\eqref{eq:bounding sums via integrals} with $a(n)=n^{-d/2}\,\PP(\sN> n-1)$ and $b(n)=e^{-d\|z\|^2/2n}$, and changing
  integration variables, yields 
  \begin{multline}
    \label{eq:pre-Lebesgue bound on D}
    \int_{0}^{\infty}s^{d/2-2}e^{-s}\left(1+\frac{2s}{d\|z\|^2}\right)^{-d/2}\PP\left(\frac{2\sN}{d\|z\|^2}>s^{-1}\right) 
    \diff s \\
    \le \frac{2\pi^{d/2}}{d} \|z\|^{d-2} D(z) \le \\ \frac{\pi^{d/2}}{d} \|z\|^{d-2} \bar{p}_1(z)
+\int_{0}^{d\|z\|^2/4}s^{d/2-2}e^{-s}\left(1-\frac{2s}{d\|z\|^2}\right)^{-d/2}
\PP\left(\frac{2\sN}{d\|z\|^2}>s^{-1}-\frac{4}{d \|z\|^2}\right)
\diff s
  \end{multline}
  In the upper bound, the $\bar{p}_1(z)$ term is treated separately since $a(t-1)b(t)$ is not integrable on $(1,\infty)$.
  
  Now consider sequences $\sN_L$, $a_L$ and $z_L$ as described in the statement of the proposition, and
  substitute $\sN=\sN_L$ and $z=z_L$ in~\eqref{eq:pre-Lebesgue bound on D}.
  Since $a_L$ is positive and non-decreasing, it either converges to a strictly positive limit, or diverges to $+\infty$.
  Consequently, since $\PP(\sN_L/a_L \le \cdot)$ converges weakly to $G$ as $L\to\infty$, standard convergence of types arguments (see
  e.g.~\cite[pp. 193]{Billingsley94}), imply that, for any fixed $c\in\RR$ and almost every $y\in\RR$, as $L\to\infty$ we have
  \begin{equation}
    \label{eq:weak convergence}
  \lim_{L\to\infty} \PP\left(\frac{\sN_L}{a_L} \le \frac{d\|z_L\|^2}{2a_L}\,y -\frac{c}{a_L}\right) = G\left(\frac{d}{2}\,\xi^2\, y\right).
  \end{equation}
  Then, since $s^{d/2-2}e^{-s}$ is integrable on $(0,\infty)$ when $d\ge3$, applying Lebesgue's dominated convergence theorem to the
  integrals in the lower and upper bounds in~\eqref{eq:pre-Lebesgue bound on D} shows, in both cases, that the limits as $L\to\infty$ exist and equal
  $$
  \int_{0}^{\infty}s^{d/2 - 2}e^{-s} [1-G(d\xi^2/2s)]\diff s.
  $$
  It then follows from~\eqref{eq:pre-Lebesgue bound on D} that
\begin{equation}
 \lim_{L\to\infty} \|z_L\|^{d-2}D(z_L) =
 \frac{d}{2\pi^{d/2}} \int_{0}^{\infty}s^{d/2 - 2}e^{-s} [1-G(d\xi^2/2s)]\diff s.
 \label{eq:D limit}
 \end{equation}

 We now consider $E_1$. Let $z\in\ZZ^d$ and $n\in\posint$. Since $\ind(z\not\leftrightarrow n)=\ind(z\leftrightarrow n+ 1)$,
 changing variables via $n\mapsto n+1$ in the first sum in~\eqref{eq:E1 definition} yields
 $$
   \label{eq:E1 sums}
   2E_1(z) \le \bar{p}_1(z) + \sum_{n=1}^\infty \left|\bar{p}_{n+1}(z) - \bar{p}_n(z)\right|,
   $$
 while changing variables the second sum yields
 $$
   2E_1(z) \ge -\left(\bar{p}_1(z) + \sum_{n=1}^\infty \left|\bar{p}_{n+1}(z) - \bar{p}_n(z)\right|\right).
 $$
 It then follows from Lemma~\ref{lem:p and pbar sum bounds} that $E_1(z_L)=O(\|z_L\|^{-d+\epsilon})$ as $L\to\infty$, for every
 $\epsilon>0$, and so
\begin{equation}
  \label{eq:E1 limit}
 \lim_{L\to\infty} \|z_L\|^{d-2}\,E_1(z_L) = 0.
\end{equation}

 Finally, now consider $E_2$. In this case, Lemma~\ref{lem:p and pbar sum bounds} immediately implies that
 $E_2(z_L)=O(\|z_L\|^{-d+\epsilon})$ as $L\to\infty$, for every $\epsilon>0$, and so
\begin{equation}
  \label{eq:E2 limit}
 \lim_{L\to\infty} \|z_L\|^{d-2}\,E_2(z_L) = 0.
\end{equation}
The stated result follows by combining~\eqref{eq:D limit}, \eqref{eq:E1 limit} and~\eqref{eq:E2 limit}.
\end{proof}

We now turn to the proof of Lemma~\ref{lem:p and pbar sum bounds}. 

 \begin{proof}[Proof of Lemma~\ref{lem:p and pbar sum bounds}]
   The local central limit theorem for random walk (see e.g.\cite[Theorem 1.2.1]{Lawler1991}) implies that there exists
   $c_1\in(0,\infty)$ such that for all $n\in\posint$ and $z\in\ZZ^d$ we have 
   \begin{equation}
     |\bar{p}_n(z) - p_n(z)| \le c_1 n^{-d/2-1}.
     \label{eq:LCLT}
   \end{equation}
   Similarly, it can be shown (see e.g.\cite[Lemma 6.1]{ZhouGrimmDengGaroni2020}) that there exists $c_2\in(0,\infty)$ such that
   for all $n\in\posint$ and $z\in\ZZ^d$ we have
   \begin{equation}
     |\bar{p}_n(z) - \bar{p}_{n+1}(z)| \le c_2 n^{-d/2-1}.
     \label{eq:bound on difference of consecutive pnbars}
   \end{equation}
   Let $a\in\posint$. It follows from~\eqref{eq:LCLT}, via~\eqref{eq:bounding sums via integrals}, that
   \begin{equation}
     \sum_{n=a+1}^\infty|\bar{p}_n(z) - p_n(z)| \le  c_1 \int_{a+1}^\infty (t-1)^{-d/2-1} \dt = \frac{2}{d}c_1 a^{-d/2}. 
     \label{eq:upper sum for pn vs pnbar}
   \end{equation}
   Similarly, it follows from~\eqref{eq:bound on difference of consecutive pnbars} that
   \begin{equation}
     \sum_{n=a+1}^\infty|\bar{p}_n(z) - \bar{p}_{n+1}(z)| \le \frac{2}{d}c_2 a^{-d/2}.
     \label{eq:upper sum for pnbars}
   \end{equation}

   Now suppose $1\le n \le a$. From~\eqref{Eq: pnbar(x)} there exists $c_3\in(0,\infty)$ such that 
   \begin{equation}
     \bar{p}_n(z), \bar{p}_{n+1}(z) \le c_3 \exp\left(-\dfrac{d\|z\|^2}{2(a+1)}\right)
     \label{eq:exponential bounds for pnbars}
   \end{equation}
   But, as shown e.g. in~\cite[Proposition 2.1.2]{LawlerLimic2010}, there exist $\beta,c_4\in(0,\infty)$ such that for all $n\in\naturals$
   and $s>0$
   $$
   \PP\left(\max_{0\le j \le n} \|S_j\| \ge s\sqrt{n}\right) \le c_4 e^{-\beta s^2}.
   $$
   It then follows that for all $1\le n \le a$ and $z\in\ZZ^d$ we have
   \begin{equation}
     p_n(z) \le c_4 \exp\left(-\frac{\beta\|z\|^2}{a+1}\right).
       \label{eq:exponential bounds for pn}
   \end{equation}
   From~\eqref{eq:exponential bounds for pnbars} and~\eqref{eq:exponential bounds for pn} we then conclude that there exist
   $c_5,\gamma\in(0,\infty)$, independent of $a$, such that for any $a\in\posint$ we have
   \begin{equation}
     \sum_{n=1}^a|\bar{p}_n(z)-\bar{p}_{n+1}(z)|, \quad \sum_{n=1}^a|\bar{p}_n(z)-p_{n}(z)|
     \le
     c_5\, a\, \exp\left(-\gamma \frac{\|z\|^2}{(a+1)}\right).
       \label{eq:lower sums for pn and pnbar bounds}
   \end{equation}
   
   Now fix $\epsilon\in(0,\infty)$ and let $z\neq 0$. Choosing $a=\lceil\|z\|\rceil^{2-\frac{2\epsilon}{d}}$ implies that
   the sums in ~\eqref{eq:lower sums for pn and pnbar bounds} are exponentially small, and combining with 
   \eqref{eq:upper sum for pn vs pnbar} and~\eqref{eq:upper sum for pnbars} then implies that for any $\epsilon\in(0,\infty)$ there exists
   $c\in(0,\infty)$ such that 
   \begin{equation}
     \sum_{n=1}^\infty|\bar{p}_n(z)-\bar{p}_{n+1}(z)|, \quad      \sum_{n=1}^\infty|\bar{p}_n(z)-\bar{p}_{n+1}(z)|
     \le c \|z\|^{-d+\epsilon}.
   \end{equation}
   Both parts of the stated result now follow by specialising to the case $z=z_L$.
  \end{proof}

 \ack{
   This research was supported by the Australian Research Council Centre of Excellence for Mathematical and Statistical Frontiers (Project
   no. CE140100049), and the Australian Research Council's Discovery Projects funding scheme (Project No. DP180100613). 
   It was undertaken with the assistance of resources and services from the National Computational Infrastructure (NCI), which is supported
   by the Australian Government. 
   Y. D. acknowledges the support by the National Key R\&D Program of China under Grant No. 2018YFA0306501 and by the National 
   Natural Science Foundation of China under Grant No. 11625522.
}

 \appendix

 \section{Appendix}
 \label{sec:Appendix A}
 \subsection{Random-length Random Walk and Random-length LERW}
 In this brief appendix we provide some details outlining how~\eqref{eq:RLRW weight}, \eqref{eq:RLLERW weight}
 and~\eqref{eq:LERW unwrapped two-point function} can be obtained.

 We begin by considering RLRW on $\ZZ^d$. Therefore, let $\rho$ be given by~\eqref{eq:RLRW weight} and, let $z\in\ZZ^d$. Then
    \begin{align*}
      \EE \sum_{n=0}^{|\sZ|} \ind(\sZ_n=z) &=\EE \sum_{n=0}^{\infty} \ind(|\sZ|\ge n) \ind(\sZ_n=z)\\
      &= \EE \sum_{n=0}^{\infty} \ind(|\sN|\ge n) \ind(S_n=z)\\
      &= \sum_{n=0}^{\infty} \PP(|\sN|\ge n) \,\PP(S_n=z)\\
      &= \sum_{n=0}^{\infty} \sum_{\substack{\omega\in\Omega_{\ZZ^d}^n \\\omega:0\to z}} \frac{\PP(\sN\ge |\omega|)}{(2d)^{|\omega|}} \\
      &= \sum_{\substack{\omega\in\Omega_{\ZZ^d}\\\omega:0\to z}}\rho(\omega)\\
  \end{align*}
 which confirms that the two-point function~\eqref{def:greens_function_rlrw} is indeed of the form~\eqref{eq:Ising correlation via walk}
 with weight~\eqref{eq:RLRW weight}. Precisely the same argument confirms the analogous statement for RLRW on the torus. 

 We now turn our attention to, $\sL$, the RLLERW on the torus. Let $\Sigma_{\torus}$ denote the subset of $\Omega_{\torus}$ consisting of
 self-avoiding walks. 
 Let $x\in\torus$. Since $\sL$ is self-avoiding, we have
 $$
   \EE \sum_{n=0}^{|\sL|}\ind(\sL_n=x)
   = \sum_{\tau\in\Sigma_{\torus}}\PP(\sL=\tau) \sum_{n=0}^{|\tau|}\ind(\tau_n=x)
   = \sum_{\tau\in\Sigma_{\torus}}\PP(\sL=\tau) \ind(\tau\ni x)
 $$
 But it can be easily shown that for any map $f:\Sigma_{\torus}\to\RR$, we have for all $x\in\torus$ that 
 \begin{equation}
   \label{eq:SAW sum decomposition}
 \sum_{\substack{\tau\in\Sigma_{\torus}\\ \tau \ni x}}f(\tau)=
 \sum_{\substack{\eta\in\Sigma_{\torus}\\ e(\eta)=x}}\,
 \sum_{\substack{\tau\in\Sigma_{\torus}\\ \tau \sqsupseteq\eta}}f(\tau)
 \end{equation}
 It then follows, in particular, that
 \begin{align*}
   \EE \sum_{n=0}^{|\sL|}\ind(\sL_n=x)
   &=  \sum_{\substack{\eta\in\Sigma_{\torus}\\ e(\eta)=x}}  \sum_{\substack{\tau\in\Sigma_{\torus}\\ \tau \sqsupseteq\eta}}\PP(\sL=\tau)\\
   &=  \sum_{\substack{\eta\in\Sigma_{\torus}\\ e(\eta)=x}} \PP(\sL \sqsupseteq\eta) \\
   &= \sum_{\substack{\eta\in\Omega_{\torus}\\ \eta:0\to x}} \rho(\eta)
 \end{align*}
 with $\rho$ given by~\eqref{eq:RLLERW weight}. 
 We conclude that the RLLERW two-point function is indeed of the form~\eqref{eq:Ising correlation via walk} with $\rho$ as
 in~\eqref{eq:RLLERW weight}.

 Finally, we now consider~\eqref{eq:LERW unwrapped two-point function}. Again let $\rho$ be given by~\eqref{eq:RLLERW weight}, and 
 let $z\in\ZZ^d$. Since $\sL$ is self-avoiding we have
 \begin{align*}
   \tilde{g}_{\rho}(z) &= \sum_{\eta\in\Omega_{\torus}} \ind[e\circ\sW^{-1}(\eta)=z]\, \PP(\sL\sqsupseteq\eta) \\
   &=\sum_{\substack{\eta\in\Sigma_{\torus}\\e\circ\sW^{-1}(\eta)=z}}
   \sum_{\substack{\tau\in\Sigma_{\torus}\\\tau\sqsupseteq\eta}}\PP(\sL=\tau) \\
 \end{align*}
 But since $\sW^{-1}(\eta)$ is self-avoiding whenever $\eta$ is, a slight variation of the argument leading to~\eqref{eq:SAW sum
   decomposition} shows that for any map $f:\Sigma_{\torus}\to\RR$ and $z\in\ZZ^d$ 
 \begin{equation}
   \label{eq:SAW sum unwrapped decomposition}
 \sum_{\substack{\eta\in\Sigma_{\torus}\\ e\circ\sW^{-1}(\eta)=z}}\,
 \sum_{\substack{\tau\in\Sigma_{\torus}\\ \tau \sqsupseteq\eta}}f(\tau)=
  \sum_{\substack{\tau\in\Sigma_{\torus}\\ \sW^{-1}(\tau) \ni z}}f(\tau)
 \end{equation}
 It then follows that
 \begin{align*}
   \tilde{g}_{\rho}(z) &= \sum_{\substack{\tau\in\Sigma_{\torus}\\\sW^{-1}(\tau)\ni z}} \PP(\sL=\tau) = \PP[\sW^{-1}(\sL)\ni z],
 \end{align*}
 as claimed in~\eqref{eq:LERW unwrapped two-point function}.
 
\section*{References}
\bibliographystyle{unsrt}

\end{document}